\documentclass[10pt]{article}
\usepackage{style}
\usepackage{subcaption}

\AtEveryBibitem{ %
    \clearfield{day}
    \clearfield{month}
    \clearfield{series}
    \clearfield{venue}
    \clearname{editor}
    \clearlist{publisher}
    \clearlist{location} %
    \clearfield{venue}
    \clearfield{issn}
    \clearfield{isbn}
    \clearfield{urldate}
    \clearfield{eventdate}
    \clearfield{pages}
    \clearfield{number}
    \clearfield{volume}
}

\makeatletter
\newcommand{\symbolthanks}[2]{%
  \begingroup
    \renewcommand{\thefootnote}{#1}%
    \footnotemark
    \protected@xdef\@thanks{%
      \@thanks
      \protect\symbolthanks@text{#1}{\the\c@footnote}{#2}}%
  \endgroup
}

\newcommand{\symbolthanks@text}[3]{%
  \begingroup
    \renewcommand{\thefootnote}{#1}%
    \footnotetext[#2]{#3}%
  \endgroup
}
\makeatother

\usepackage{setspace}
\newcommand{\AES}{\mathsf{AES}}

\newcommand{\sgn}{\operatorname{sgn}}

\allowdisplaybreaks

\newcommand{\PSPACE}{\mathsf{PSPACE}} 
\newcommand{\RE}{\mathsf{RE}} 
\newcommand{\AM}{\mathsf{AM}} 
 
\newcommand{\QIP}{\mathsf{QIP}} 
\newcommand{\NEXP}{\mathsf{NEXP}} 
\newcommand{\MIP}{\mathsf{MIP}}
\newcommand{\QMIP}{\mathsf{QMIP}}
\newcommand{\QCMA}{\mathsf{QCMA}}
\newcommand{\QMA}{\mathsf{QMA}}
\newcommand{\BQP}{\mathsf{BQP}}
\newcommand{\SBQP}{\mathsf{SBQP}}

\renewcommand{\eqref}[1]{\text{eq.~}(\ref{#1})}

\title{A quantum oracle separation between $\QMA(2)$ and $\QMA$}
\author{
John Bostanci\thanks{Simons Institute \& UC Berkeley. \texttt{cbostanc@berkeley.edu}} 
\and Sabee Grewal\thanks{IBM Research, Columbia University. \texttt{sabee@ibm.com}} 
\and Jonas Haferkamp\thanks{Ruhr-University Bochum. \texttt{jonas.haferkamp@rub.de}}
\and Andrew Huang\thanks{University of California, Berkeley. \texttt{a\_huang@berkeley.edu}}
\and Yeongwoo Hwang\symbolthanks{\ensuremath{\Delta}}{Harvard University. \texttt{yeongwoohwang@g.harvard.edu}}
\and Anand Natarajan\thanks{MIT. \texttt{anandn@mit.edu}}
\and Chinmay Nirkhe\thanks{University of Washington. \texttt{nirkhe@cs.washington.edu}}
}

\date{}

\begin{document}

\maketitle

\begin{abstract}
    We find a quantum oracle relative to which $\QMA\neq\QMA(2)$. As a consequence, we resolve the no-disentanglers conjecture of Watrous:    
    for every $\eps+\delta<1$, any $(\eps,\delta)$-disentangler requires input size exponential in the number of output qubits.
    Our proof combines the unitarily invariant polynomial method of She and Yuen (ITCS~'23) with a new construction based on the symmetric and antisymmetric subspace projectors, reducing the $\QMA$ lower bound to the approximate degree of $\mathrm{OR}$.
\end{abstract}

\newif\iftoc\tocfalse

\iftoc
\newpage
\fi

\iftoc
\hypersetup{linktocpage}
\setcounter{tocdepth}{2}
{\tableofcontents}
\setlength{\parskip}{4pt}%
\newpage
\fi

\section{Introduction}
Quantum information can fundamentally alter the structure and power of proof systems.
Two major results illustrate this especially clearly. 
First, quantum interactive proofs can be parallelized to three messages, $\QIP=\QIP(3)$~\cite{KW00}, whereas an analogous parallelization for classical interactive proofs would imply $\PSPACE = \AM$, collapsing $\PSPACE$ to the second level of the polynomial hierarchy.
Second, and perhaps even more dramatically, $\MIP=\NEXP$ while $\QMIP = \MIP^*=\RE$, so allowing entanglement between provers increases the power of multiprover proofs all the way to undecidable problems~\cite{BFL91,JNV+21}.

In this work, we study a more mysterious quantum resource: \emph{unentanglement}. 
The class $\QMA(2)$ consists of problems that can be verified efficiently given two quantum proofs that are promised to be unentangled with one another~\cite{KMY01}. 
Whereas entanglement is usually viewed as a source of quantum power, $\QMA(2)$ asks whether the \emph{absence} of entanglement can improve verification for untrusted quantum systems.

Understanding the power of $\QMA(2)$ is a major open problem in quantum complexity theory and is closely connected to fundamental questions about entanglement, separability, and optimization; see the recent survey of Jeronimo, Wu, and Leigh~\cite{JWL26}. 
Despite 25 years of study, at the level of a general complexity-class containment, essentially all that is known is $\QMA\subseteq\QMA(2)\subseteq\NEXP$,
and there is surprisingly little evidence to suggest whether $\QMA(2)$
is ``closer'' in power to $\QMA$ or to $\NEXP$.

A basic test of whether a promise of unentanglement improves verification is whether it can yield an advantage in the black-box setting.
Can a $\QMA(2)$ verifier solve a quantum oracle problem efficiently that no $\QMA$ verifier can?
Such a separation would isolate unentanglement as a genuine source of computational power in proof systems: the two proof systems are given identical black-box access to the problem and differ only in the structure imposed on their proofs.
This question has remained open for nearly two decades~\cite{AK07,ABD+08},
and Aaronson~\cite{Aar21} highlighted it as a fundamental question in quantum query complexity.
We resolve one version\footnote{A classical oracle separation still eludes us!} of this question by giving the first unitary-oracle separation between
$\QMA$ and $\QMA(2)$.
\begin{theorem}[Informal]
\label{thm:oracle-separation-informal}
There exists a unitary oracle $\calO$ such that
\[
\QMA^\calO \neq \QMA(2)^\calO\,.
\]
\end{theorem}

In particular, we introduce a family of black-box problems that a
$\QMA(2)$ verifier can solve using a single query and linear-size proofs,
whereas any $\QMA$ verifier must use either exponentially many queries or
an exponentially large proof.

The same $\QMA$ query lower bound also has consequences beyond the oracle setting.
Specifically, it gives strong lower bounds on approximate
\emph{disentanglers}: quantum channels whose outputs are always close to
separable states and whose images approximately contain every separable state. 
Watrous observed (as reported in~\cite{ABD+08}) that in the query model a disentangler with only polynomially many input qubits would allow $\QMA(2)$ to be simulated by $\QMA$. 
Our query separation thus implies that no such disentanglers exist, as Watrous conjectured.\footnote{Aaronson et al.~\cite{ABD+08} formulated a stronger quantitative conjecture, suggesting a
lower bound of $\dim(\mathcal H)=2^{\Omega(\dim(\mathcal K))}$ for a channel from $\mathcal H$ to $\mathcal K\otimes\mathcal K$.
Recently, Jeronimo, Wu, and Xu~\cite{JWX26}
refuted this quantitative conjecture by constructing, for every fixed
$\eps\in(0,1)$, an $(\eps,0)$-disentangler satisfying
\[
\dim(\mathcal H) = \exp\!\left( O_\eps\!\left( \sqrt{\dim(\mathcal K)}\log\dim(\mathcal K) \right) \right)\,.
\]}

\begin{theorem}[Informal]
\label{thm:disentangler-informal}
For every fixed $\eps,\delta\geq0$ with $\eps+\delta<1$, an
$(\eps,\delta)$-disentangler requires exponentially many input qubits in
the number of output qubits.
\end{theorem}

Before our work, unconditional lower bounds for general approximate
disentanglers were known only in much more restricted parameter regimes.
Aaronson et al.~\cite{ABD+08} ruled out perfect $(0,0)$-disentanglers in
every finite dimension, while Harrow, Natarajan, and Wu~\cite{HNW19}
obtained a quasipolynomial lower bound on the input dimension when $\eps+\delta$ is exponentially small in the number of output qubits.
For constant error, Akibue, Kato, and Tani~\cite{AKT25} proved an
exponential lower bound for the more restrictive class of
\emph{strong} disentanglers.

\pdfbookmark[1]{Acknowledgments}{acknowledgments}
\paragraph*{Acknowledgments.}
SG is supported by the Herman Goldstine Memorial Postdoctoral Fellowship. AN is supported by a Sloan Fellowship and NSF CAREER award 2339948. CN is supported by NSF CAREER award IIS2541127. YH is supported by NSF Award Nos. 2238836, 2430375, and an IBM PhD Fellowship. This work was initiated [in part] while the author(s) were visiting the Simons Institute for the Theory of Computing, supported by NSF QLCI Grant No. 2016245.

\pdfbookmark[1]{Tool and computational resource disclosure}{tool-disclosure}
\paragraph*{Tool and computational resource disclosure.}
A research manuscript should record not only its findings, but also how those findings were obtained. For AI-assisted research, we believe such disclosure gives reviewers relevant context for evaluating the work, provides transparency about the research process as the community develops norms around the use of these tools, and helps establish an accurate understanding of their capabilities and limitations.

Accordingly, we disclose that the proof idea underlying the main theorem was generated using ChatGPT 5.6 Sol. 
Our initial prompts directed the model to the work of She and Yuen~\cite{SY22}, and, through subsequent prompting with minimal additional guidance, the model proposed the proof idea presented in this paper.
The authors subsequently verified, simplified, and developed the argument presented here, and take full responsibility for its correctness, and exposition.

\section{Overview}
Any computational power of $\QMA(2)$ over $\QMA$ must be derived from its soundness guarantee: a $\QMA(2)$ verifier may assume that even a dishonest proof is a product state. 
This suggests a natural unitary query problem for separating the two classes.
Consider a black-box unitary that marks a subspace, with the promise that either the marked subspace contains a product state, or every product state has small overlap with it. 
A $\QMA(2)$ verifier can solve such a problem trivially: 
any witness is guaranteed to be product, so it need only test whether the received state lies in the marked subspace. 
For a $\QMA$ verifier, by contrast, this test need not be sound, since the witness may be an arbitrary entangled state.

The main challenge in implementing this idea is to construct a family of marked subspaces for which every query-efficient verifier fails on some instance. %
Our approach combines the unitary polynomial method of She and
Yuen~\cite{SY22}---an extension of the standard polynomial method of
Beals et al.~\cite{BBC+01}---with a highly symmetric family
of marked subspaces. 
We view the identification of this family, together with the 
``miraculous'' symmetry properties that drastically simplify the polynomials arising in She and Yuen's method, as the main technical innovation of this work. 

\subsection{A starting point: guessing the witness}

A useful starting point is to eliminate the witness in
$\QMA$ by the standard technique of amplifying and guessing the
witness. Suppose a query problem has a $T$-query $\QMA$ algorithm with
completeness $2/3$ and soundness $1/3$, with a witness consisting of $m$
qubits. By applying the in-place amplification procedure of
Marriott-Watrous, this implies an $O(Tm)$-query $\QMA$ algorithm with
completeness $1 - 2^{-2m}$ and soundness $2^{-2m}$, that still takes a
witness of $m$ qubits. 

We can now eliminate the witness by running this verifier on a \emph{uniformly random} witness.
In the YES case, the acceptance probability is at least $2^{-m}(1 - 2^{-2m}) \geq
2^{-(m+1)}$, and in the NO case, it is at most $2^{-2m}$.
Although both acceptance-probability thresholds are exponentially close, they are still far apart in a \emph{multiplicative} sense: their ratio is exponentially large.
What we will show is that there is no witness-free quantum algorithm
making polynomially many queries whose acceptance probability
satisfies this multiplicative-error condition. (Technically, what we are doing is using the containment $\QMA \subseteq \SBQP$, and showing a lower bound against $\SBQP$.) We will show this by
applying a version of the polynomial method, as we will describe in
the next section.

Before going on, it is worth noting that the technique of guessing the
witness does \emph{not} apply to $\QMA(2)$. In fact, the first step of
in-place amplification breaks down: it is known~\cite{ABD+08} that any in-place
amplification procedure for $\QMA(2)$ would imply that $\QMA(2)
\subseteq \QMA$, and this argument holds even in the oracle setting. Thus, our oracle separation implies that any in-place amplification for $\QMA(2)$ in the plain model, if it exists, must be non-relativizing.

\subsection{The polynomial method}

The polynomial method of Beals et al.~\cite{BBC+01} starts from the observation that the acceptance probability of a quantum algorithm making $T$ queries to a classical oracle $(x_1, \dots, x_N)$ is a degree-$2T$ multivariate polynomial $p(x_1, \dots, x_N)$ in the oracle bits. 
Thus, to prove a query lower bound, it suffices to show that no low-degree polynomial has the acceptance behavior of the quantum algorithm. 
Since multivariate polynomials are hard to reason about, it is often crucial in these arguments to find symmetries in the problem that let us reduce the number of variables as much as possible---ideally to one.

The prototypical example of this method in action is the $\mathsf{BQP}$ query lower bound for the $\mathrm{OR}$ problem (also known as the unstructured search problem). 
The problem is to distinguish the NO case, in which every entry of the oracle is $0$, from the YES case, in which at least one entry is $1$. 
This problem is invariant under permutations of the oracle bits: whether the input $(x_1, \dots, x_N)$ is a YES or NO instance depends only on its \emph{Hamming weight} $x_1 + \dots + x_N$.
To apply this symmetry, let us consider the \emph{averaged} acceptance probability of the algorithm over all instances of Hamming weight $r$:
\begin{equation} f(r) := \E_{(x_1, \dots, x_N) \in \{0,1\}^N: \sum_i x_i = r} p(x_1, \dots, x_N)\,. \label{eq:symmetrization-grover} \end{equation}
Basic results in invariant theory imply that $f$ itself must be a
degree-$2T$ univariate polynomial in $r$. 

To obtain a query lower
bound, observe that the success conditions require $f(0)$ to be close to $0$, while $f(1),\dots,f(N)$ must all be close to $1$.
This means the graph of $f$ must make a 
sharp ``jump'' near $r =0$, followed by a long flat ``runway.'' 
A low-degree polynomial cannot have this shape, and standard bounds imply that in fact $f$ must have degree at least $\Omega(\sqrt{N})$. 

The same picture can be adapted to $\QMA$ using the amplify-and-guess argument mentioned above. 
Consider, for instance, the complement of $\mathrm{OR}$, where $r=0$ is the YES case and every $r \ge 1$ is the NO case.\footnote{Note that the argument we are about to discuss does not apply to $\mathrm{OR}$, which is actually in $\QMA$! The reason is that, after rescaling, the ``runway'' for the OR problem need not remain flat: the rescaled polynomial is allowed to take any value from around $1$ up to $2^m$ on these points.}
Let $f(r)$ now denote the averaged acceptance probability of the corresponding witness-guessing algorithm --- i.e., the witness-free algorithm that accepts YES instances with probability $\ge 2^{-m}(1-2^{-2m})$ and accepts NO instances with probability $\le 2^{-2m}$. Consider the
rescaled polynomial $g \coloneqq 2^m f$.
Then $g(0) \ge 2/3$ (and in fact possibly as large as $2^m$), while $g(r) \in [0, 2^{-m}]$ for $r = 1, 2,
\dots, N$. 
Thus the long NO runway remains exponentially close to zero after
rescaling, while the isolated YES point remains bounded away from zero.
This again forces a large polynomial degree.

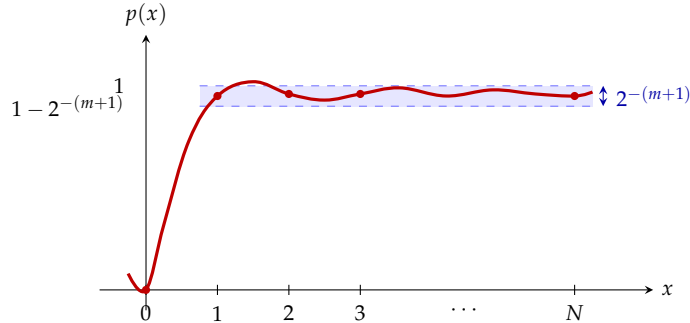
\begin{figure}[ht]
\begin{mdframed}
\caption{\textbf{The polynomial degree lower bound we use.}
Consider a polynomial $p:\R\to\R$ satisfying
$p(0)=0$ and $1-2^{-(m+1)}\leq p(j)\leq1$
for every $j\in\{1,\ldots,N\}$.
Standard approximate-degree bounds for $\mathrm{OR}$ imply that such a
polynomial has degree
$\Omega(\min\{N,\sqrt{Nm}\})$.}
\label{fig:poly-graph}

\begin{center}
\begin{tikzpicture}[
scale=0.9,
transform shape,
x=1.05cm,
y=3cm,
>=stealth,
every node/.style={font=\small}
]
    \def\runwaylow{0.90}
    \def\runwayhigh{1.00}

    \fill[blue!10]
        (0.75,\runwaylow) rectangle (6.25,\runwayhigh);
    \draw[blue!55, dashed]
        (0.75,\runwaylow) -- (6.25,\runwaylow);
    \draw[blue!55, dashed]
        (0.75,\runwayhigh) -- (6.25,\runwayhigh);

    \draw[->] (-0.65,0) -- (7.10,0)
        node[right] {\(x\)};
    \draw[->] (0,-0.10) -- (0,1.25)
        node[above] {\(p(x)\)};

    \foreach \x/\lab in {
         0/{0},
         1/{1},
         2/{2},
         3/{3},
         6/{N}
    } {
        \draw (\x,-0.025) -- (\x,0.025);
        \node[below=3pt] at (\x,0) {\(\lab\)};
    }
    \node[below=3pt] at (4.45,0) {\(\cdots\)};

    \draw[red!75!black, very thick]
        plot[smooth, tension=0.65] coordinates {
            (-0.25,0.08)
            ( 0.00,0.00)
            ( 0.25,0.32)
            ( 0.60,0.73)
            ( 1.00,0.95)
            ( 1.50,1.02)
            ( 2.00,0.96)
            ( 2.50,0.93)
            ( 3.00,0.96)
            ( 3.55,0.99)
            ( 4.20,0.95)
            ( 4.85,0.98)
            ( 5.45,0.96)
            ( 6.00,0.95)
            ( 6.25,0.97)
        };

    \fill[red!75!black] (0,0) circle (1.7pt);

    \foreach \x/\y in {
        1/0.95,
        2/0.96,
        3/0.96,
        6/0.95
    } {
        \fill[red!75!black] (\x,\y) circle (1.7pt);
    }

    \node[left=3pt] at (-0.08,\runwayhigh) {\(1\)};
    \node[left=3pt] at (-0.08,\runwaylow)
        {\(1-2^{-(m+1)}\)};

    \draw[<->, blue!65!black]
        (6.40,\runwaylow) -- (6.40,\runwayhigh)
        node[midway, right=2pt] {\(2^{-(m+1)}\)};

\end{tikzpicture}
\end{center}
\end{mdframed}
\end{figure}

The idea of applying the polynomial method template to \emph{unitary} query lower bounds, and to the $\QMA$ vs $\QMA(2)$ problem in particular,
was introduced by She and Yuen~\cite{SY22}.
For ease of exposition, we
explain their method in the case where the unitary oracle is Hermitian, which is the only case we will need. 
She and Yuen started by observing that for any quantum algorithm making $T$ queries to a
\emph{unitary} oracle $U$, its acceptance probability is given by a
polynomial $p(U)$ of total degree $2T$ in the entries of $U$. 
They then identified a symmetry that is natural in the $\QMA(2)$ setting:
suppose the oracle $U$ acts on a \emph{bipartite} Hilbert space $\CC^d
\otimes \CC^d$. 
Then it is natural to ask the $\QMA(2)$ algorithm to
decide properties of $U$ that are \emph{invariant under conjugation by local unitaries}: 
\[
U\text{ is a YES instance} \quad \iff \quad (V \otimes W) U
(V^\dagger \otimes W^\dagger) \text{ is a YES instance}.
\]
For instance, suppose $U$ is a reflection that marks a subspace, and consider the problem of distinguishing whether the marked subspace contains a product state or has small overlap with every product state. This property has the desired local-unitary symmetry and is also easily decidable by a $\QMA(2)$ proof system: simply ask for a product state in the marked subspace as the witness. This is the \emph{Entangled Subspace problem} studied by She and Yuen.

She and Yuen identified an analog of \cref{eq:symmetrization-grover} for local unitary invariance. 
For any oracle $U$, let $g(U)$ be the distribution over oracles obtained by conjugating $U$ by independent Haar-random local unitaries. 
Then the acceptance probability of a quantum algorithm, averaged over oracles drawn from $g(U)$, can be written as a linear combination of certain low-degree polynomials in $U$:
\begin{equation}
    \E_{\tilde{U} \sim g(U)} p(\tilde{U}) = \sum_{t,\sigma,\tau} \alpha_{t,\sigma,\tau} \tr[(R_\sigma \otimes R_\tau) U^{\otimes t}]\,. \label{eq:symmetrization-unitary}
\end{equation}
Here, the argument of each trace is an operator on $t$ copies of $\CC^d \otimes \CC^d$, viewed as a tensor product of $t$ copies of the ``left'' system and $t$ copies of the ``right'' system. The operators $R_\sigma$ and $R_\tau$ are permutations of the $t$ left and right registers, respectively.

While \cref{eq:symmetrization-unitary} is a significant improvement over working with all the entries of $U$, it is still much less clean than the analogous result for the $\mathrm{OR}$ problem in \cref{eq:symmetrization-grover}, which
yields a univariate polynomial. 
The main technical contribution of this work is to
identify a highly symmetric family of marked subspcaes for which all of the trace expressions in \cref{eq:symmetrization-unitary} are captured by low-degree univariate polynomials in a single paramter, the dimension $d$ of the subspace. 
Even more importantly, the same polynomial will simultaneously describe both the YES and NO instances of our promise problem. 
This will make our problem formally analogous to the $\mathrm{OR}$ problem, almost immediately yielding a $\QMA$ query lower bound.

\subsection{\texorpdfstring{Our oracle problem and the $\QMA$ lower bound}{Our oracle problem and the QMA lower bound}}
\label{sec:into_oracle_problem}

Our oracle family is a restriction of the Entangled Subspace problem
of She and Yuen to a particular set of subspaces.
We take the oracle $U$ to be a reflection $I-2\Pi$ about a subspace defined by a projector $\Pi$, with YES instances corresponding to rank-one product subspaces and NO instances corresponding to subspaces with low overlap with every product state.
She and Yuen suggested that a family of hard-to-distinguish instances could be constructed by taking random subspaces as NO instances and planting a product state to obtain YES instances.
Our innovation is to consider a more structured family of instances that lets us make better use of local-unitary invariance.

Let $D=2^{\poly(n)}$ be the local ambient dimension, so that the marked
subspace lies in $\CC^D\otimes\CC^D$.
We will construct $\Theta(\sqrt D)$ different NO instances, together with a family of YES instances, with the following two properties:
\begin{itemize}
\item Every YES instance is a rank-one product projector, whereas every
NO instance has overlap at most $1/3$ with every product state.
\item For every local-unitary invariant polynomial arising from
\cref{eq:symmetrization-unitary}, its values on all of these YES and
NO instances are simultaneously captured by a single low-degree
univariate polynomial.
\end{itemize}
The construction underlying these two properties will be described in
the next subsection.

These two facts, together with \cref{eq:symmetrization-unitary}, mean
that we can apply the polynomial method. Suppose there is a $T$-query
$\QMA$ verifier using an $m$-qubit witness. By the amplify-and-guess
argument above, and then averaging over local changes of basis, we
obtain a local-unitary invariant polynomial of degree $O(Tm)$---where
the factor $m$ comes from in-place amplification---whose value is
noticeably large on YES instances and exponentially smaller on NO
instances.

By the second property above, this invariant polynomial becomes a univariate polynomial.
After a simple change of variables, normalization, and flipping the polynomial, we obtain a polynomial $p$ such that $p(0)=0$, while $1-2^{-(m+1)}\leq p(j)\leq1$ at $\Theta(\sqrt D)$ consecutive positive integer values of $j$.
This is exactly the shape shown in \cref{fig:poly-graph}.
The polynomial degree lower bound therefore gives
\[
Tm
\geq
\Omega\!\left(\min\{\sqrt D,D^{1/4}\sqrt m\}\right) \implies 
T^2m=\Omega(\sqrt D).
\]
Thus, when $D$ is exponential in the input length, any $\QMA$ verifier must use either superpolynomially many queries or a superpolynomial-size witness.
By a standard diagonalization argument, we conclude that there exists a unitary oracle $U$ such that $\QMA^U\neq\QMA(2)^U$.

\subsection{Negative dimensions}
In constructing our families of YES and NO instances, the main
ingredients we use are projectors onto the symmetric and antisymmetric
subspaces of bipartite Hilbert spaces of the form $\mathbb{C}^d \otimes \mathbb{C}^d$, where we will refer to the two tensor factors as the ``left'' and ``right'' systems. Since these two subspaces are the eigenspaces of the SWAP operators with eigenvalue $+1$ and $-1$, respectively, in this section we use the notation $P_d^\pm$
to denote the symmetric and antisymmetric subspace projectors. These are defined as
\[ P_d^{\pm} = \frac{I \pm \SWAP}{2}\,. \]
These fulfill the required properties of YES and NO instances; the symmetric subspace contains many product states, whereas any product state has bounded overlap with the antisymmetric subspace.

By the symmetrization argument sketched in the previous section, the acceptance probability of a $T$-query quantum algorithm, averaged over the local changes of basis, is a linear combination of terms of the form
\[ \tr[ (R_{\sigma} \otimes R_{\tau}) (P_d^{\pm})^{\otimes t}]\,, \]
where $t$ ranges from $0$ to $2T$, and $R_\sigma$ and $R_\tau$ are separate permutations of the ``left'' and ``right'' systems. To carry out the polynomial method analysis from the previous section, we would like following ``dream identity'' to hold: for every $\sigma, \tau, t$, 
\begin{equation} \tr[(R_\sigma \otimes R_\tau) (P_d^{\pm})^{\otimes t}] \overset{?}{=} f_{\sigma, \tau, t}(\pm d)\,, \label{eq:not-doubled-up} \end{equation}
where $f_{\sigma, \tau, t}$ is a polynomial of degree $t$, depending on $\sigma, \tau$, and $t$. Crucially, we want the symmetric and antisymmetric projectors to be captured by the \emph{same} polynomial, evaluated at positive and negative values of the argument, respectively. Or, more pithily, we want the antisymmetric projector  to behave like a symmetric projector in a ``negative dimensional space.'' It turns out the dream identity is not always true, but it will hold for a class of $\sigma, \tau$ that will turn out to be sufficient.

To get a very basic intuition for why the dream identity might be expected to hold, we can look at the simplest such identity, the trace: $\tr[P_d^{\pm}]$. To compute the trace, we expand the projector in terms of the identity and SWAP operators:
\[ \tr[P_d^{\pm}] = \frac{1}{2} \left( \tr[I] \pm \tr[\SWAP]\right) = \frac{1}{2} (d^2 \pm d) = \frac{1}{2} ((\pm d)^2 + (\pm d))\,. \]
We see that this expression has exactly the desired form: it is a polynomial that, when evaluated at $+d$, yields the value for $P_{d}^{+}$, and when evaluated at $-d$, yields the value for $P_{d}^{-}$.

Essentially the same strategy will serve us in the general case: simply expand the tensor product of projectors in terms of identity and swap operators, to get
\[ \tr[ (R_{\sigma} \otimes R_{\tau}) (P_d^{\pm})^{ \otimes t}] = \frac{1}{2^t} \sum_{S \subseteq [t]} (\pm 1)^{|S|} \mathrm{Tr}\Big[\underbrace{(R_\sigma \otimes R_\tau) \SWAP^S}_{\substack{= \text{some permutation $R_\xi$}\\\text{of the subsystems}}}\Big]\,, \]
where we have introduced the shorthand notation $\SWAP^S$ to refer to the operator that acts as a tensor product of SWAPs on the pairs of left and right systems indexed by elements of the set $S$. Each of the individual traces in this sum is a trace of some permutation $R_{\xi}$ acting on the $2t$ registers of the quantum system. The value of such a trace is determined entirely by the \emph{cycle decomposition} of $\xi$ and the dimension $d$: in fact,
\[ \tr[R_\xi] = d^{\cyc(\xi)}\,,\]
where $\cyc(\xi)$ is the number of cycles in the cycle decomposition
of $\xi$. A very intuitive way to see this is with Penrose's
diagrammatic notation for tensors, as illustrated in
Figure~\ref{fig:trace-tensor}. 
In the diagram, every cycle of $\pi$ is converted into a closed loop upon taking the trace, which yields a factor of the dimension. Plugging this back into the above, we get
\[ \tr[(R_\sigma \otimes R_\tau) (P^{\pm}_{d})^{\otimes t }] = \frac{1}{2^t} \sum_{S \subseteq [t]} (\pm 1)^{|S|} d^{\cyc(\xi)}\,, \]
where we have notationally suppressed the dependence of $\xi$ on $S, \sigma, \tau$ to avoid clutter.

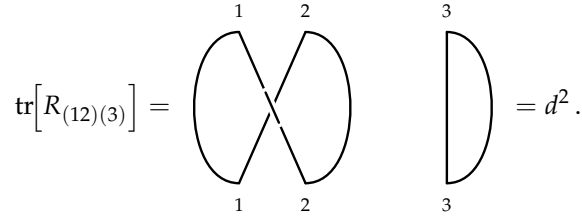
\begin{figure}[h]
\centering
\[
\tr\!\left[R_{(12)(3)}\right]
=
\vcenter{\hbox{%
\begin{tikzpicture}[
  x=1.1cm,
  y=1cm,
  wire/.style={draw, line width=0.9pt, line cap=round},
  leglabel/.style={font=\scriptsize}
]
  \coordinate (t1) at (0,1);
  \coordinate (t2) at (0.8,1);
  \coordinate (t3) at (2.5,1);
  \coordinate (b1) at (0,-1);
  \coordinate (b2) at (0.8,-1);
  \coordinate (b3) at (2.5,-1);

  \draw[wire] (t1) .. controls (-0.72,1) and (-0.72,-1) .. (b1);
  \draw[wire] (t2) .. controls ( 1.52,1) and ( 1.52,-1) .. (b2);
  \draw[wire] (t3) .. controls ( 3.22,1) and ( 3.22,-1) .. (b3);

  \draw[wire] (b1) -- (t2);
  \draw[wire] (b2) -- (t1);

  \draw[
    wire,
    preaction={draw=white, line width=3.5pt, line cap=round}
  ]
    (0.48,-0.20) -- (0.32,0.20);

  \draw[wire] (b3) -- (t3);

  \node[leglabel, above=2pt] at (t1) {$1$};
  \node[leglabel, above=2pt] at (t2) {$2$};
  \node[leglabel, above=2pt] at (t3) {$3$};
  \node[leglabel, below=2pt] at (b1) {$1$};
  \node[leglabel, below=2pt] at (b2) {$2$};
  \node[leglabel, below=2pt] at (b3) {$3$};
\end{tikzpicture}%
}}
=
d^{2}\,.
\]
\caption{A tensor network diagram illustrating
  $\tr[R_{(12)(3)}]=d^2$.}
\label{fig:trace-tensor}
\end{figure}

Now, at this point we would have shown the ``dream identity'' if the sign factor $(\pm 1)^{|S|}$ were equal to $(\pm 1)^{\cyc(\xi)}$. This is unfortunately not true in general, but \emph{is} true when $\sigma, \tau$ are both \emph{even} permutations. Under this assumption, we have
\[ \tr[(R_\sigma \otimes R_\tau) (P_d^\pm)^{\otimes t}] = \frac{1}{2^t} \sum_{S \subseteq [t]} (\pm 1)^{|S|} d^{\cyc(\xi)} = \frac{1}{2^t} \sum_{S \subseteq [t]} (\pm 1)^{\cyc(\xi)} d^{\cyc(\xi)} = \frac{1}{2^t} \sum_{S \subseteq [t]}  (\pm d)^{\cyc(\xi)}\,, \]
as desired.

It turns out one final trick lets us entirely remove the restriction on $\sigma$ and $\tau$ and derive the dream identity for a different set of projectors: define $\Pi_{d}^\pm:= (P_d^{\pm})^{\otimes 2}$ to be projectors on two copies of $\mathbb{C}^{d} \otimes \mathbb{C}^d$, viewed as $\mathbb{C}^{d^2} \otimes \mathbb{C}^{d^2}$ by ``merging'' the two left-hand systems and the two right-hand systems. Now, any permutation of $t$ merged $d^2$-dimensional registers automatically becomes an \emph{even} permutation of the $2t$ underlying $d$-dimensional registers and
\[ \tr[(R_\sigma \otimes R_{\tau}) (\Pi^\pm_d)^{\otimes t}]= \tr[\underbrace{(R_{\sigma'} \otimes R_{\tau'})}_{\text{even permutations}} (P^\pm_d)^{\otimes 2t}] = f_{\sigma', \tau', 2t}(\pm d)\,. \]

This doubled construction already captures the key intuition we need: by grouping registers into pairs, the relevant permutations become even, and the symmetric and antisymmetric projectors can be described by the same polynomial after formally negating the dimension. 
It would therefore already suffice to derive the qualitative separation $\QMA^U\subsetneq \QMA(2)^U$. However, in the main body of the paper, we use a slightly different construction that gives stronger quantitative bounds and makes the final promise problem cleaner.
Specifically, in \cref{sec:antisymmetric-projectors} we will instead take $\Pi^-_d = \Pi^{d,4}_\mathrm{anti}$, the projector onto the 4-copy antisymmetric subspace, viewed as a bipartite operator by grouping copies $1,2$ and $3,4$. This solves the sign issue in essentially the same way as the doubled projectors, since every permutation of the merged registers becomes an even permutation on the underlying systems. Moreover, this choice is ``stronger'' in the sense that $\Pi^{d,4}_\mathrm{anti} \preceq P_d^{-} \otimes P_d^{-}$, yielding an improved soundness guarantee and, ultimately, better parameters for our disentangler result.

For either construction, each projector $\Pi_d^\pm$ acts on $\C^{d^2} \otimes \C^{d^2}$, so we can embed it into a common ambient space $\C^D \otimes \C^D$ whenever $d^2 \le D$.
For the construction used in the main body, the formal point $d=-1$ corresponds to the symmetric projector $\Pi_1^+$, which is a rank-one product projector and hence gives our YES instances. For the NO instances, we take the antisymmetric projector at  $d=4,9,14,\ldots$ up to $\sqrt D$, giving $\Theta(\sqrt D)$ different NO instances.
Under the simple change of variables $d=5j-1$, the formal
YES point $d=-1$ maps to $j=0$, while the NO instances map to the
consecutive positive integers $j=1,2,\ldots$. These are exactly the
points forming the long ``runway'' in \cref{fig:poly-graph}.

\subsection{Disentanglers}

Our work has immediate consequences for the existence of \emph{approximate disentanglers}. 
\begin{definition}[Approximate disentangler]
\label{def:disentangler}
Let $\Lambda: \mathcal{B}(\C^{2^m}) \rightarrow \mathcal{B}(\C^D\otimes\C^D)$ be a quantum channel.
We say that $\Lambda$ is an $(\eps,\delta)$-disentangler if:
\begin{enumerate}
    \item for every input state $\rho$, there is a separable state $\sigma$
    such that $\frac12\norm{\Lambda(\rho)-\sigma}_1\leq\eps$ and,
    \item for every separable state $\sigma$, there is an input state $\rho$
    such that $ \frac12\norm{\Lambda(\rho)-\sigma}_1\leq\delta$.
\end{enumerate}
\end{definition}

Indeed, suppose we had an $(\epsilon, \delta)$-disentangler $\Lambda$, with $m = m(D)$. Then, in the query model\footnote{Since we are working in the query model, we ignore the computational efficiency of the disentangler: all that matters is that it takes no queries to implement it.}, any $\QMA(2)$ protocol whose completeness-soundness gap is larger than $\eps + \delta$ can be simulated by a $\QMA$ protocol. To see this, let us take a $\QMA(2)$ protocol with completeness $c$ and soundness $s$, where the verifier expects two unentangled proofs of $w$ qubits each. A $\QMA$ verifier for the same problem is to apply the disentangler $\Lambda$ on a witness of $m(2^w)$ qubits, then simulate the $\QMA(2)$ verifier on the output of the disentangler. This yields a $\QMA$ verifier with completeness at least $c - \eps$ (the prover sends a state whose image, under the disentangler, is $\eps$-close to the optimal $\QMA$ witness), and soundness at most $s + \delta$ (any state the prover sends gets mapped to a state that is $\delta$-close to separable). Moreover, this $\QMA$ verifier makes the same number of queries as the original $\QMA(2)$ verifier.

By taking the contrapositive of this simulation argument, our separation between $\QMA$ and $\QMA(2)$ will rule out the existence of disentanglers where $m(D)$ is small. To get quantitative results, we start with a more precise version of our separation: we show that for our version of the Entangled Subspace problem, for every $D'$ that is a perfect square, the following statements hold.
\begin{enumerate}
    \item A single-query $\QMA(2)$ algorithm solves the problem on oracle instances acting on the space $\CC^{D'} \otimes \CC^{D'}$, with a witness consisting a product state in this space, completeness $c = 1$, and constant soundness $s < 1$.
    \item Any $\QMA$ algorithm with $m$-qubit witnesses on these instances requires $\Omega(\sqrt{(D')^{1/2}/m})$ queries.
\end{enumerate}
Putting these together, suppose we have an $(\eps, \delta)$-disentangler with $\eps + \delta < c- s$. Then the simulation yields a $\QMA$ algorithm making $1$ query and taking a witness of size $m(D')$. The $\QMA$ lower bound therefore implies that

\[
    m(D') \geq \Omega((D')^{1/2})\,.
\]
In the regime $\epsilon + \delta \in [0, c-s)$, this lower bound matches the upper bound obtained by \cite{JWX26}, up to a $\log(D)$ factor. With a little bit of additional work, we can extend the result to hold for $\epsilon + \delta > c-s$, with a weaker quantitative lower bound.

\section{Preliminaries}

We write $[n]\coloneqq\{1,\ldots,n\}$. For a permutation $\pi$, we write $\cyc(\pi)$ to denote the number of cycles in its cycle decomposition. Let $U(d)$ denote the unitary group on $\C^d$.
Whenever we average over $U(d)$, it is with respect to the Haar measure.
We use standard definitions of $\QMA$ and $\QMA(2)$; see, e.g.,
Watrous~\cite{Watrous09}.
We also assume familiarity with the basics of quantum computing and quantum information; see, e.g., \cite{watrous2018theory}.

Next, we define the product value of a bipartite operator, which measures
its maximum overlap with a product state.

\begin{definition}[Product value]
\label{def:prod-val}
For a positive semidefinite operator $P$ on a bipartite Hilbert space
$A \otimes B$, define its \emph{product value} by
\[
h_{\rm Sep}(P)
\coloneqq
\max_{\substack{\ket{\alpha} \in A,\ \ket{\beta} \in B\\
\norm{\alpha}=\norm{\beta}=1}}
\bra{\alpha}\bra{\beta} P \ket{\alpha}\ket{\beta}\,.
\]
\end{definition}

We also use the following well-known lower bound on the degree of polynomial approximations to $\mathrm{OR}$.
\begin{theorem}[Approximate Degree Bound for OR, \cite{BCWZ99}]\label{thm:approximate_degree_bound}
    Let $q: \RR \to \RR$ be a real polynomial of degree $d$. If $q(0) = 0$ and $1-\eps \leq q(i) \leq 1$ for all $1 \leq i \leq N$, then $d \geq \Omega(\min\{N, \sqrt{N \log(1/\eps)}\})$. 
\end{theorem}

\subsection{The antisymmetric and symmetric subspaces}

Let $S_k$ denote the symmetric group on $k$ elements.  
For $\pi\in S_k$, let $R_\pi$ be the unitary operator on $(\mathbb C^d)^{\otimes k}$ that permutes the $k$ subsystems according to $\pi$:
\[
R_\pi
\bigl(v_1\otimes\cdots\otimes v_k\bigr) = v_{\pi^{-1}(1)}\otimes\cdots\otimes v_{\pi^{-1}(k)}\,.
\]

\begin{definition}[Antisymmetric subspace]
The \emph{antisymmetric subspace} of $(\mathbb C^d)^{\otimes k}$ is
\[
\wedge^k \mathbb C^d
\coloneqq
\left\{
\ket\psi \in(\C^d)^{\otimes k} :
R_\pi|\psi\rangle
=
\sgn(\pi)|\psi\rangle
\text{ for every }\pi\in S_k
\right\}\,.
\]
The orthogonal projector onto the antisymmetric subspace is defined as
\begin{equation}
\Pi_{\mathrm{anti}}^{d,k} \coloneqq \frac{1}{k!} \sum_{\pi\in S_k} \sgn(\pi)R_\pi\,.
\label{eq:antisymmetric-projector}
\end{equation}
We refer to $\Pi_{\mathrm{anti}}^{d,k}$ as the \emph{antisymmetric projector}.
\end{definition}

\begin{definition}[Symmetric subspace]
The \emph{symmetric subspace} of $(\C^d)^{\otimes k}$ is
\[
\vee^k \C^d \coloneqq \left\{ \ket\psi\in(\C^d)^{\otimes k} : R_\pi\ket\psi
= \ket\psi \text{ for every }\pi\in S_k \right\}\,.
\]
The orthogonal projector onto the symmetric subspace is defined as
\begin{equation}
\Pi_{\mathrm{sym}}^{d,k} \coloneqq \frac{1}{k!} \sum_{\pi\in S_k} R_\pi\,.
\label{eq:symmetric-projector}
\end{equation}
We refer to $\Pi_{\mathrm{sym}}^{d,k}$ as the \emph{symmetric projector}.
\end{definition}

\subsection{Local-unitary invariant polynomials}

Let $A=\C^{d_A}$ and $B=\C^{d_B}$. A polynomial $f$ in the entries of an
operator $X$ on $A\otimes B$ is \emph{local-unitary invariant} if
\[
f\left((U\otimes V)X(U\otimes V)^\dagger\right)=f(X)
\]
for every $U\in U(d_A)$ and $V\in U(d_B)$.

We will use the following standard characterization of these invariant
polynomials~\cite{Bra37,Pro76,SY22}.

\begin{theorem}
\label{thm:lu-invariant-generators}
For $t\geq1$ and $\pi,\sigma\in S_t$, let $R_\pi^A$ and $R_\sigma^B$
denote the operators that permute the $t$ copies of $A$ and $B$,
respectively. Then every homogeneous degree-$t$ local-unitary invariant
polynomial is a linear combination of the polynomials
\[
X \longmapsto \tr\left[ X^{\otimes t}
\left(R_\pi^A\otimes R_\sigma^B\right) \right], \qquad \pi,\sigma\in S_t\,.
\]
\end{theorem}

\section{On the antisymmetric and symmetric projectors}
\label{sec:antisymmetric-projectors}

We study the antisymmetric and symmetric projectors acting on four systems under the bipartition $(12\mid34)$.
For convenience, we write $\Pi_d^-\coloneqq \Pi_{\rm anti}^{d,4}$ and $\Pi_d^+\coloneqq \Pi_{\rm sym}^{d,4}$. 

We prove two properties that drive the oracle separation. 
First, we show that the antisymmetric projector is a good candidate for NO instances of the entangled subspace problem, i.e., it has small overlap with every product state:
\[
h_{\rm Sep}(\Pi_d^-)\leq \frac 1 3\,.
\]
Second, for every fixed local-unitary invariant polynomial $G$, we show that there is a single univariate polynomial $g$ such that
\[
g(d)=G(\Pi_d^-)\,, \qquad g(-d)=G(\Pi_d^+)\,.
\]
Intuitively, this means we will be able to apply the polynomial method to quantum algorithms that switch between querying the symmetric and antisymmetric subspaces simultaneously.  Our first property shows that switching the sign of the subspace corresponds to changing between YES and NO instances of the entangled subspace problem, which will give us our degree lower bound.

\subsection{Small overlap with product states}
\label{subsec:prod-value}

We compute the product value of the antisymmetric projector across the bipartition $(12\mid 34)$.

\begin{lemma}
\label{lem:product-value}
For every $d \geq 4$, $h_{\rm Sep}(\Pi_d^-) \leq \frac{1}{3}$.
\end{lemma}

\begin{proof}
 
Since $\Pi_d^- = \Pi_d^- \bigl( \Pi_{\rm anti}^{d,2}\otimes\Pi_{\rm anti}^{d,2} \bigr)$,
it suffices to consider pairs of unit vectors in the antisymmetric subspace, $\ket{\alpha},\ket{\beta}\in\wedge^2\C^d$.
Set $\ket{\psi}\coloneqq \ket{\alpha}_{12}\ket{\beta}_{34}$.

By definition, the antisymmetric subspace projector can be written as $\sum_{\tau} \sgn(\tau) R_{\tau}$.  However, consider the subgroup $H=\{e,(12),(34),(12)(34)\}$, i.e., permutations that only swap within the two parts of $\ket{\psi} = \ket{\alpha} \ket{\beta}$.
Since every $h\in H$ applies a permutation to a state in the antisymmetric subspace, by definition we have $R_h\ket{\psi}=\sgn(h)\ket{\psi}$. 
Therefore, for any permutation $r\in S_4$,
\[
\sgn(rh)R_{rh}\ket{\psi} = \sgn(r)R_r\ket{\psi}\,.
\]
This means that for any pair of permutations $\pi, \pi'$ such that $\pi = \pi' \circ \sigma$ for $\sigma \in H$, the terms $\sgn(\pi) R_{\pi}$ and $\sgn(\pi') R_{\pi'}$ act the same on $\ket{\psi}$.
There are $\abs{S_4}/\abs{H} = 6$ cosets of $H$, and we choose representatives $e,(13),(14),(23),(24),(13)(24)$ from each coset to obtain our expression. 
Thus, when restricted to $\wedge^2\C^d\otimes\wedge^2\C^d$, the four-system antisymmetric
projector has the following form:
\begin{equation}
\label{eq:asym_expr}
\Pi_d^- = \frac{1}{6} \left(I-R_{13}-R_{14}-R_{23}-R_{24}+R_{13}R_{24}\right)\,.
\end{equation}
Here, we use $R_{ij}$ denote the swap of registers $i$ and $j$.

We can now evaluate $\bra \psi \Pi^-_d \ket \psi$ given this reduced form. Let $\rho_\alpha$ and $\rho_\beta$ denote the one-register reduced states of $\ket{\alpha}$ and $\ket{\beta}$, respectively. 
By antisymmetry, the two one-register marginals of each state are equal, so the choice of register traced out is irrelevant.

For every $i\in\{1,2\}$ and $j\in\{3,4\}$, the swap identity $\tr(R(A\otimes B)) = \tr(AB)$ gives
\begin{equation*}
\bra{\psi}R_{ij}\ket{\psi}
= \tr\left[R_{ij}(\rho_\alpha \otimes \rho_\beta) \right]
= \tr(\rho_\alpha\rho_\beta)\,.
\end{equation*}
Because $\rho_\alpha$ and $\rho_\beta$ are positive semi-definite, $\tr(\rho_\alpha\rho_\beta)\geq 0$, and all four single-swap terms in \Cref{eq:asym_expr} are
nonnegative. The remaining term is also simple to evaluate. The operator
$R_{13}R_{24}$ exchanges the two pairs of registers, so $R_{13}R_{24} \ket{\alpha}_{12}\ket{\beta}_{34} = \ket{\beta}_{12}\ket{\alpha}_{34}$.
Therefore,
\[
\bra{\psi}R_{13}R_{24}\ket{\psi}
=
\abs{\braket{\alpha}{\beta}}^2
\leq 1\,.
\]
Substituting these identities into \Cref{eq:asym_expr} gives
\begin{align*}
6\bra{\psi}\Pi_d^-\ket{\psi} &=
1 - 4\tr(\rho_\alpha\rho_\beta) + \abs{\braket{\alpha}{\beta}}^2  \leq 2\,.
\end{align*}
Hence, we have that for all product states $\ket{\psi}$,
\[
\bra{\psi}\Pi_d^-\ket{\psi} \leq \frac{1}{3}\,.
\]
Taking the maximum over product states proves the claim.
\end{proof}

\subsection{Local-unitary invariants of the symmetric and antisymmetric projectors}
\label{subsec:lu-antisym}

We next study the behavior of $\Pi_d^-$ and $\Pi_d^+$ under local-unitary invariant polynomials. 
In particular, we show that for every fixed local-unitary invariant
polynomial $G$, there is a univariate polynomial $g$ such that $g(d)=G(\Pi_d^-)$ and $g(-d)=G(\Pi_d^+)$.

To show this, it suffices by \cref{thm:lu-invariant-generators} to consider the spanning family of invariants defined as follows.
For the bipartition $(12\mid34)$, let $A = B =(\C^d)^{\otimes 2}$.
For an operator $P$ on a bipartite Hilbert space $A\otimes B$, and for $t\geq1$ and $\pi,\sigma\in S_t$, define
\[
J_{\pi,\sigma}(P) \coloneqq \tr\left[ P^{\otimes t} \left(R_\pi^A\otimes R_\sigma^B\right) \right]\,,
\]
where $R_\pi^A$ and $R_\sigma^B$ permute the $t$ copies of $A$ and $B$, respectively.

\begin{lemma}
\label{lem:formal-dimension}
Fix $t\geq1$ and $\pi,\sigma\in S_t$.
There is a polynomial $q_{\pi,\sigma}\in\R[x]$ of degree at most $4t$
such that, for every integer $d\geq4$,
\[
q_{\pi,\sigma}(d) = J_{\pi,\sigma}(\Pi_d^-)\,,
\]
and, for every integer $d \geq 1$,
\[
q_{\pi,\sigma}(-d) = J_{\pi,\sigma}(\Pi_d^+)\,.
\]
\end{lemma}
\begin{proof}
	Let $\xi_{\pi, \sigma}\in S_{4t}$ be the permutation of the $4t$ underlying systems induced by $R_\pi^A\otimes R_\sigma^B$.
Since $A$ and $B$ each consist of two systems and $\pi$ and $\sigma$ permute two-system blocks,
$\xi_{\pi, \sigma}$ is even. For $\bold{\tau}=(\tau_1,\ldots,\tau_t)\in S_4^t$, write
$\overline \tau \coloneqq \tau_1\oplus\cdots\oplus\tau_t \in S_{4t}$.

We expand $(\Pi_d^-)^{\otimes t}$ as $\sum_{\tau \in S^t_4} \sgn(\overline\tau) R_{\overline\tau}$ and apply this identity to get 
\[
J_{\pi,\sigma}(\Pi_d^-) 
= \tr\left[(\Pi_d^-)^{\otimes t} R_{\xi_{\pi, \sigma}}\right] 
= \frac{1}{(4!)^t} \sum_{\bold\tau \in S_4^t} \sgn(\overline\tau)\tr(R_{\overline\tau {\xi_{\pi, \sigma}}})\,.
\]
Observe that $\tr(R_\rho)=d^{\cyc(\rho)}$. 
To see this, first note that the trace of $R_\rho$ is the number of standard basis vectors that it fixes. 
A standard basis vector is fixed by $R_\rho$ exactly when its indices are constant on each cycle of $\rho$, giving $d^{\cyc(\rho)}$ such basis vectors. Continuing the above derivation, this yields
\[
J_{\pi,\sigma}(\Pi_d^-) = \frac{1}{(4!)^t} \sum_{\bold\tau \in S_4^t} \sgn(\overline\tau)\tr(R_{\overline\tau {\xi_{\pi, \sigma}}})= \frac{1}{(4!)^t} \sum_{\bold\tau \in S_4^t} \sgn(\overline\tau) d^{\cyc(\overline \tau {\xi_{\pi, \sigma}})} =: q_{\pi,\sigma}(d)\,.
\]
Since $\cyc(\overline\tau{\xi_{\pi, \sigma}})\leq 4t$, we have $\deg(q_{\pi,\sigma})\leq 4t$.

Until now, we defined the polynomial $q_{\pi, \sigma}$ using only the antisymmetric subspace projector, but it has an expression that only depends on $d$, the input.  Surprisingly, we will see that when we plug in a negative number for $d$, we get the same expression as if we had defined $q_{\pi, \sigma}$ using the symmetric subspace projector. Notice that for every $\rho\in S_{4t}$, we have $ (-d)^{\cyc(\rho)} = \sgn(\rho)d^{\cyc(\rho)}\,. $
Indeed, suppose the disjoint-cycle decomposition of $\rho$ has cycle lengths $\ell_1,\ldots,\ell_{\cyc(\rho)}$. Then each cycle of length $\ell_i$ has sign $(-1)^{\ell_i-1}$. Since $\sum_{i=1}^{\cyc(\rho)}\ell_i=4t$,
we obtain
\[
\sgn(\rho) = \prod_{i=1}^{\cyc(\rho)}(-1)^{\ell_i-1} = (-1)^{4t-\cyc(\rho)} = (-1)^{\cyc(\rho)}\,,
\]
where the last equality uses that $4t$ is even.

Plugging this back into the definition of $q_{\pi, \sigma}$, we have
\[
q_{\pi,\sigma}(-d)
= \frac{1}{(4!)^t} \sum_{\boldsymbol{\tau}\in S_4^t} \sgn(\overline\tau) \sgn(\overline\tau{\xi_{\pi, \sigma}}) d^{\cyc(\overline\tau{\xi_{\pi, \sigma}})}
= \frac{1}{(4!)^t} \sum_{\boldsymbol{\tau}\in S_4^t} d^{\cyc(\overline\tau{\xi_{\pi, \sigma}})}\,,
\]
where the second equality uses $\sgn(\overline\tau)\sgn(\overline\tau{\xi_{\pi, \sigma}}) = \sgn({\xi_{\pi, \sigma}}) = 1$.

Remarkably, the last expression is exactly the expansion obtained from
$(\Pi_d^+)^{\otimes t}$, since the symmetric projector contains no
permutation signs. Hence
\[
q_{\pi,\sigma}(-d) = J_{\pi,\sigma}(\Pi_d^+)\,.\qedhere
\]
\end{proof}

\cref{lem:formal-dimension} shows that the basic local-unitary invariants of $\Pi_d^-$ and $\Pi_d^+$ are captured by \emph{the same} low-degree polynomial, with the symmetric family appearing at negative values of the dimension.
We now extend this correspondence to arbitrary local-unitary invariant polynomials by expanding them into the local-unitary basis polynomials.

In the following lemma, we consider algorithms that can query the projectors $\Pi_{d}^{\pm}$ for
different values of $d$.  When this happens, we imagine that the projectors act on the first $d$
dimensions of a pair of $\sqrt{D}$-dimensional subspace, and act as $0$ elsewhere.  Since we will
eventually twirl by a pair of Haar random unitaries, the choice of which subspace the projectors act on will not matter.

\begin{lemma}
\label{lem:lu-univariate-reduction}
Let $G$ be a real, local-unitary invariant polynomial of degree at most $L$ on operators acting on $\C^D\otimes\C^D$. 
Then there exists a polynomial $g\in\R[x]$ of degree at most $4L$ such that for every integer $4 \leq d \leq \sqrt{D}$,
\[
g(d)=G(\Pi_d^-)\,,
\]
and for every integer $1 \leq d \leq \sqrt{D}$,
\[
g(-d)=G(\Pi_d^+)\,.
\]
In particular, for every rank-one product projector $\psi \otimes \phi$ on
$\C^D\otimes\C^D$, $g(-1)=G(\psi \otimes \phi)$.
\end{lemma}

\begin{proof}
Write $G=\sum_{t=0}^L G_t$, where $G_t$ is homogeneous of degree $t$. 
Since $G$ is local-unitary invariant, each $G_t$ is also local-unitary invariant.

For $t\geq1$, \cref{thm:lu-invariant-generators} says we can express $G_t$ as a
linear combination of the basis polynomials
\[
J_{\pi,\sigma}(P)
=
\tr\left[
P^{\otimes t}
\left(R_\pi^A\otimes R_\sigma^B\right)
\right],
\qquad
\pi,\sigma\in S_t\,.
\]
\cref{lem:formal-dimension} tells us that each of these basis polynomials can be written as a univariate polynomial in the dimension $d$.  When we substitute those in, we find that there is a polynomial
$g_t$ in the dimension $d$, of degree at most $4t$, such that $g_t(d)=G_t(\Pi_d^-)$
for every $4 \leq d \leq \sqrt{D}$, and
$g_t(-d)=G_t(\Pi_d^+)$ for every $1 \leq d \sqrt{D}$.
For $t=0$, $G_0$ is constant, so we take $g_0=G_0$.

Therefore, $g\coloneqq\sum_{t=0}^L g_t$ has degree at most $4L$ and satisfies
\[
g(d)=G(\Pi_d^-) \qquad \text{and} \qquad g(-d)=G(\Pi_d^+)
\]
for every valid $d$.

Finally, note that every rank-one product projector on $\C^D\otimes\C^D$ is locally unitarily
equivalent to $\Pi_1^+$, so by the local-unitary invariance of $G$, $g(-1)=G(\psi \otimes \phi)$ for
every rank-one product projector $\psi \otimes \phi$.
\end{proof}
\section{The oracle separation}
\label{sec:oracle-separation}

We present the unitary-oracle separation between $\QMA$ and $\QMA(2)$. The separation is based on the antisymmetric-symmetric correspondence from \cref{sec:antisymmetric-projectors}.

\begin{restatable}{theorem}{oraclesep}
\label{thm:oracle-separation}
There exists a unitary oracle $\calO$ such that
\[
\QMA^\calO\neq\QMA(2)^\calO.
\]
\end{restatable}

We begin by defining the promise problem underlying the separation. Let 
\begin{equation}
	\label{eq:n_choice}
	N = N(D) = \max \{j \in \mathbb N\ :\ (5j - 1)^2 \leq D\}\,.
\end{equation}
We will consider $d_j \coloneq 5j - 1$ for each $j \in [N]$. \footnote{The choice of $d_j$ is so that the affine map $d=5x-1$ sends $x=0$ to the formal symmetric point $d=-1$, while sending $x=j$ to the antisymmetric dimension $d_j$.  In particular, $d_1=4$ is the smallest dimension for which $\Pi_d^-$ is nonzero.}
Recall that $\Pi_d^-$ acts on the bipartite space $(\C^d)^{\otimes2}\otimes(\C^d)^{\otimes2}$, whose
local dimension is $d^2$. By choice of $N$, $\Pi_{d_1}^-,\ldots,\Pi_{d_N}^-$ can act on as projectors on $\C^D\otimes\C^D$. 
We fix such embeddings and, in a slight abuse of notation, continue to write $\Pi_{d_j}^-$ for the resulting projectors.

We define a family of yes and no instances for the entangled subspace problem, and we denote the problem of deciding between yes and no instances from this family the antisymmetric entangled subspace problem.

\begin{definition}[Antisymmetric Entangled Subspace problem ($\AES$)]
\label{def:aes}
Fix a rank-one product projector $\psi \otimes \phi$ on $\C^D\otimes\C^D$, and define
\begin{align*}
\calY_{D} &\coloneqq \Braces{ (U\otimes V)\psi \otimes \phi(U\otimes V)^\dagger : U,V\in U(D) }\,,\\
\calN_{D} &\coloneqq \bigcup_{j=1}^{N(D)} \Braces{ (U\otimes V)\Pi_{d_j}^-(U\otimes V)^\dagger : U,V\in U(D) }\,.
\end{align*}
For a projector $P$ on $\C^D\otimes\C^D$, define the unitary oracle
\begin{equation}
\label{eq:projector-oracle}
\calO_P \coloneqq \proj{0}\otimes I + \proj{1}\otimes(I-2P)\,.
\end{equation}
The \emph{Antisymmetric Entangled Subspace problem}, denoted
$\AES_{D}\coloneqq(\calY_D,\calN_D)$, is the promise problem
of deciding, given oracle access to $\calO_P$, whether
$P\in\calY_D$ or $P\in\calN_D$.
\end{definition}

By definition, YES instances are rank-one product projectors on
$\C^D\otimes\C^D$, while the NO instances are locally unitarily
equivalent to one of $\Pi_{d_1}^-,\ldots,\Pi_{d_N}^-$.
Since local unitaries and zero-padding preserve product value, our YES instances have $h_{\rm Sep}(P) = 1$, and 
\cref{lem:product-value} implies that every NO instance satisfies 
$h_{\rm Sep}(P)\leq\frac{1}{3}$.

The promise problem $\AES_D$ is a restriction of the Entangled Subspace problem of She and Yuen~\cite{SY22}. 
Recall that in the $(a,b)$-Entangled Subspace problem, a YES subspace contains a state that is $a$-close in trace distance to a product state, whereas in a NO subspace, every state is at least $b$-far from every product state.
The problem $\AES_D$ is a special case of the $(0,2/3)$-Entangled Subspace problem.
An interesting distinction between the general entangled subspace problem and the antisymmetric subspace problem is that our YES instances only have to consist of projectors onto rank one subspaces.
On the other hand, every NO instance is at least $2/3$-far from
every product state, as follows from \cref{lem:product-value}.

\subsection{The \texorpdfstring{$\QMA(2)$}{QMA(2)} upper bound}

The $\QMA(2)$ upper bound for the Entangled Subspace problem of
She and Yuen~\cite{SY22} already applies to our promise problem.
For $\AES_D$, however, the verification is particularly simple,
so we include it for completeness.

\begin{proposition}
\label{prop:qma2-upper-bound}
$\AES_D$ has a one-query $\QMA(2)$ verifier with
completeness $1$ and soundness at most $1/3$.
\end{proposition}
\begin{proof}
The $\mathsf{QMA}(2)$ verifier will receive a copy of a product state and measure the projector $P$. For completeness, we briefly remind the reader that the two-outcome
measurement $\{I-P,P\}$ can be done with a single query to $\mathcal{O}_{P}$ via the Hadamard test: 
On input $\ket{\psi}$, prepare the
control qubit in $\ket{+}$, query $\calO_P$, and applying a Hadamard
to the control register to get the state
\[
\ket{0} \otimes (I-P)\ket{\psi} + \ket{1}\otimes P\ket{\psi}. \]
Then, measuring the control qubit measures $\{P, I-P\}$. 

If $P\in\calY_D$, then the range of $P$ is spanned by a product state,
which the two provers can send to make the verifier accept with probability $1$.

If $P\in\calN_D$, the projector $P$ is
locally unitarily equivalent to $\Pi_{d_j}^-$ for some $j \in [N(D)]$. 
Hence, by \cref{lem:product-value},
\[
h_{\rm Sep}(P) = h_{\rm Sep}(\Pi^-_{d_j}) \leq \frac 1 3 \,,
\]
where we also use that the zero-padding $\Pi^{-}_{d_j}$ into $\C^D\otimes\C^D$ does not
change the product value. Therefore every product witness is accepted
with probability at most $1/3$.
\end{proof}

\subsection{The \texorpdfstring{$\QMA$}{QMA} lower bound}

We now prove that any $\QMA$ verifier for $\AES_D$ must either use many oracle queries or receive a large witness.
Our proof has two main steps. On one hand, we apply the guessing lemma~\cite{aaronson2011impossibility} (or in-place amplification~\cite{marriott2005quantum}) to find a low-degree polynomial that approximates the maximum acceptance probability of the verifier.  Then, we use the symmetric--antisymmetric correspondence from \cref{subsec:lu-antisym} to lower bound the degree of this polynomial by the approximate degree of $\mathrm{OR}$.

Our main theorem is as follows.

\begin{theorem}
\label{thm:qma-lower-bound}
Fix constants $0\leq s<c\leq1$.
Suppose a $\QMA$ verifier decides $\AES_D$ with completeness $c$ and soundness $s$, using an $m$-qubit witness and $q$ queries to $\calO_P$. Then
\begin{equation*}
    q^2(m+1) = \Omega(\sqrt D)\,.
\end{equation*}
\end{theorem}

We begin by converting the maximum acceptance probability of a $\QMA$ verifier into a low-degree polynomial.
Fix a verifier $V$.
For each projector $P$ specifying an oracle $\calO_P$, let $M_P$ denote its acceptance operator on the witness register, so that
\[
\Pr[V^{\calO_P}\text{ accepts }\rho] = \tr(M_P\rho)
\]
for every witness state $\rho$.
The maximum acceptance probability of $V$ is thus
\[
\lambda(P)\coloneqq\lambda_{\max}(M_P)\,.
\]
Since $\calO_P$ has degree at most one in the entries of $P$, the entries of $M_P$ are low-degree polynomials in the entries of $P$.
The obstacle is that taking the largest eigenvalue does not preserve this polynomial structure.
We therefore use trace powers of $M_P$, which retain the polynomial dependence on $P$ while still corresponding to its largest eigenvalue.

\begin{lemma}[Trace power method, or the guessing lemma~\cite{aaronson2011impossibility, marriott2005quantum}]
\label{lem:qma-trace-polynomial}
For every positive integer $\ell$, the function
\[
H_\ell(P)\coloneqq \tr(M_P^\ell)
\]
is a real polynomial of degree at most $2q\ell$ in the entries of $P$, and satisfies
\[
\lambda(P)^\ell \leq H_\ell(P) \leq 2^m\lambda(P)^\ell\,.
\]
\end{lemma}
\begin{proof}
Since $\calO_P$ has degree at most one in the entries of $P$ and $V$ makes
$q$ oracle queries, every entry of $M_P$ is a polynomial of degree at most
$2q$. Hence $H_\ell(P)=\tr(M_P^\ell)$ has degree at most $2q\ell$. If $\lambda_1,\ldots,\lambda_{2^m}$ are the eigenvalues of $M_P$, then
$H_\ell(P) = \sum_{i=1}^{2^m}\lambda_i^\ell$.
Since $\lambda(P)=\max_i\lambda_i$, the claimed bounds follow immediately.
\end{proof}

Although the above lemma is phrased in terms of a function $H_{\ell}$, we note that this function exactly implements the guessing lemma from \cite{aaronson2011impossibility}, in the sense that $M_{P}^{\ell}$ is the result of applying the in-place amplification algorithm of Marriot and Watrous~\cite{MW04}, and taking the trace corresponds to inputting a uniformly random witness (up to scaling by $2^{m}$).  Thus, this function $H_{\ell}$ actually corresponds to the accepting probability of a $q\ell$-query quantum algorithm making queries to $P$.   

We are now ready to prove the main theorem. 

\begin{proof}[Proof of \cref{thm:qma-lower-bound}]
Let $V$ denote the $\QMA$ verifier. Pick a sufficiently large $\ell = O(m+1)$ such that $2^m\left(s/c\right)^\ell \leq 2^{-(m+1)}$. Let $H\coloneqq H_\ell$ be the polynomial from \cref{lem:qma-trace-polynomial}.
If $P$ is a YES instance, then $\lambda(P) \geq  2/3$, and hence
\[
H(P)\geq c^\ell.
\]
If $P$ is a NO instance, then $\lambda(P) \leq 1/3$, and hence
\[
0\leq H(P) \leq 2^m s^\ell \leq 2^{-(m+1)} c^\ell\,.
\]
Note also that $\deg(H)\leq 2q\ell$.
However, $H$ is not necessarily local-unitary invariant. We find a local-unitary invariant polynomial by averaging $H$ over local unitaries. Define
\[
\overline H(P) \coloneqq \E_{U,V\in U(D)} H\!\left( (U\otimes V)P(U\otimes V)^\dagger \right)\,,
\]
where the expectation is over the Haar measure.
Then $\overline H$ is a real local-unitary invariant polynomial of degree
at most $2q\ell$. To see this, first note that conjugation by
$U\otimes V$ is linear in $P$, and therefore cannot increase the degree
of $H$. The local-unitary invariance follows from the fact that Haar
measure is invariant under left and right multiplication by a fixed
unitary.

Because applying local unitaries maps yes instances to yes instances and no instances to no instances, the bounds we found for the value of $H$ also apply to $\overline{H}$. In particular,
\[
\overline H(E)\geq c^\ell\,,
\]
while for every $j\in[N]$, with $N = N(D)$,
\[
0\leq \overline H(\Pi^-_{d_j}) \leq 2^{-(m+1)} c^\ell\,.
\]

By \cref{lem:lu-univariate-reduction}, there is a univariate polynomial
$h$ of degree at most $8q\ell$ such that $h(-1)=\overline H(E)$ and $h(d_j)=\overline
H(\Pi^-_{d_j})$  for every $j\in[N]$.
To apply \cref{thm:approximate_degree_bound}, we need to normalize and flip the polynomial so that it goes between $0$ and a value close to $1$ (as opposed to $c^{\ell}$ and $2^{-(m+1)} c^{\ell}$).  To that end, define
\[
p(x) \coloneqq 1-\frac{h(5x-1)}{h(-1)}\,.
\]
Since $h(-1)=\overline H(E)\geq c^\ell>0$, this is well defined.
Moreover, we have that $p(0) = 0$, and for every $j \in [N]$,
    \[ 1-2^{-(m+1)} \leq p(j) \leq 1\,. \]
Thus, by \cref{thm:approximate_degree_bound}, $\deg(p) = \Omega\left(\sqrt{N(m+1)}\right)$.
On the other hand, $\deg(p)\leq 8q\ell$.
Since $\ell=O(m+1)$, we conclude that $q^2(m+1) = \Omega(N)$. Finally, recalling from
\Cref{eq:n_choice} that $N = \floor{\tfrac 1
5(\sqrt D -1)}$, this yields
\[
q^2(m+1) = \Omega(\sqrt D)\,. \qedhere
\]
\end{proof}

The oracle separation \Cref{thm:oracle-separation} now follows by relatively standard arguments. For completeness, we include a proof of the diagonalization argument in the appendix.

\section{The no-disentanglers theorem}
We derive the no-disentanglers theorem as a consequence of the
$\QMA$ lower bound from \cref{thm:qma-lower-bound}. We rule out polynomial-size disentanglers throughout the full regime $\eps+\delta<1$. In particular, if each output register contains $n$ qubits, then any $(\eps,\delta)$-disentangler requires $2^{\Omega_{\eps,\delta}(n)}$ input qubits.

At a high level, the idea is that a good disentangler allows us to construct a $\mathsf{QMA}$ algorithm for the $\mathsf{AES}_{D}$ problem by using the disentangler and then calling the $\QMA(2)$ algorithm. We will see that the parameters $\epsilon$ and $\delta$ eat into the completeness / soundness gap of the $\QMA(2)$ verifier, and thus there are slight complications when $\epsilon + \delta$ become larger than $2/3$; we include a proof in the appendix which mitigates this issue via parallel amplification of the $\QMA(2)$ verifier.
\begin{lemma}
\label{lem:qma_disentangler_completeness_soundness}
    Assume that there exists a $\mathsf{QMA}(2)$ verifier taking a pair of witnesses with dimension $D'$, with completeness $c$ and soundness $s$, for the $\mathsf{AES}_{D}$ problem, making $t$ queries to the oracle.  If exists a $(\epsilon, \delta)$-disentangler, $\Lambda$, with output dimension $D'$ and $m$ input qubits, then there exists a $\mathsf{QMA}$ algorithm for the $\mathsf{AES}_{D}$ problem making $t$-queries, taking a witness of size $m$, with completeness $c - \delta$ and soundness $s + \epsilon$.  
\end{lemma}
\begin{proof}
    The $\mathsf{QMA}$ algorithm simply applies the disentangler to their $m$-qubit witness and then calls the $\mathsf{QMA}(2)$ verifier on the resulting state and its oracle.  By the guarantee of the $\mathsf{QMA}(2)$ verifier, if $\mathcal{O}$ is a YES instance of the $\mathsf{AES}_{D}$ problem, there exists a pair of states $\sigma_1 \otimes \sigma_2$ that causes it to accept with probability $c$.  By the definition of a $(\epsilon, \delta)$-disentangler, there is a state $\rho$ on $m$-qubit such that $\Lambda(\rho)$ is $\delta$-close to $\sigma_1 \otimes \sigma_2$ in trace distance.  On this state, the $\mathsf{QMA}$ verifier accepts with probability $c - \delta$.  

    Similarly, if $\mathcal{O}$ is a NO instance of the $\mathsf{AES}_{D}$ problem, then for all states $\sigma_1 \otimes \sigma_2$, the $\mathsf{QMA}(2)$ verifier accepts with probability at most $s$.  By the definition of a $(\epsilon, \delta)$-disentangler, for all $\rho$ given to the $\mathsf{QMA}$ verifier, there exists $\sigma_1 \otimes \sigma_2$ that is $\epsilon$-close in trace distance to $\Lambda(\rho)$.  Therefore, the $\mathsf{QMA}$ verifier accepts with probability at most $s + \epsilon$.  
\end{proof}

From this, we can immediately get an almost tight bound on disentangler that satisfy $\epsilon + \delta < 2/3$.  
\begin{corollary}
    Fix constants $\epsilon, \delta \geq 0$ such that $\epsilon + \delta < 2/3$.  Then every $(\epsilon, \delta)$-disentangler with output dimension $D$ requires 
    \begin{equation*}
        m = \Omega_{\epsilon, \delta}(\sqrt{D})\,.
    \end{equation*}
\end{corollary}
\begin{proof}
    Note that the $\mathsf{QMA}(2)$ verifier from \Cref{prop:qma2-upper-bound} has completeness $1$ and soundness $1/3$, makes $1$ query, and has input dimension $D$.  Thus, the existence of an $(\epsilon, \delta)$-disentangler and \Cref{lem:qma_disentangler_completeness_soundness} implies the existence of a $\mathsf{QMA}$ verifier that also makes $1$ query, takes as input a $m$-qubit witness, and has completeness $1 - \delta$ and soundness $1/3 + \epsilon$.  Since $\epsilon$ and $\delta$ are constants, we can apply \Cref{thm:qma-lower-bound} to get a lower bound of $m = \Omega(\sqrt{D})$, as desired.  
\end{proof}

In \Cref{sec:improved_disentangler} we apply parallel amplification to get a bound in the general $(\epsilon, \delta)$ setting.

\begin{restatable}{corollary}{disentangleramp}
\label{thm:no-disentanglers-amp}
Fix constants $\eps,\delta\geq 0$ such that $2/3 < \eps+\delta<1$ and let $\eta = 1 - \epsilon - \delta$.  Then, every $(\epsilon, \delta)$-disentangler with output dimension $D$ requires
\[
m
=
\Omega_{\eps,\delta}\left(D^{1/\lceil \log_{4/3}(2/\eta)\rceil}\right)
\]
input qubits.
\end{restatable}

To summarize the results in this section, when $\eps+\delta<2/3$, we obtain
\[
m=\Omega_{\eps,\delta}(\sqrt D)\,.
\]
This is essentially optimal (up to a $\log(D)$ factor) in light of recent work of Jeronimo, Wu, and
Xu~\cite{JWX26}. For every fixed
$\eps\in(0,1)$, they construct an $(\eps,0)$-disentangler with local
output dimension $D$ using
\[
m=O_\eps(\sqrt D\log D)
\]
input qubits. 
When $\epsilon + \delta > 2/3$, we get a lower bound that scales as $D^{1/\log(\eta^{-1})}$, where $\eta$ is the difference between $1$ and $\epsilon + \delta$.  This scaling might be the correct form for small $\eta$, but our bound is almost certainly not tight up to constants. The proof could be improved to achieve a slightly improved bound (in some regimes) of $D^{1/\left(\lceil 1/\eta\rceil - 1\right)}$ by modifying the construction to use antisymmetric subspace projectors on different sized registers; however, this is also probably not tight. We leave the task of finding the optimal upper and lower bounds on disentanglers as interesting future work.

\section{Conclusion}
\label{sec:conclusion}

\paragraph{What was needed to obtain this result. } Our paper's moral starting point is the work of She and Yuen~\cite{SY22}. There, the authors made the observation that for any unitarily invariant property, a $q$-query $\QMA$ verifier reduces to the polynomial,
\begin{equation}
    \label{eq:conc_sy_form}
    F(X) = \sum_{t = 0}^{2q} \sum_{\sigma, \tau \in S_t} \tr[(R_\pi \otimes R_\sigma)X^{\otimes t} ]\,.
\end{equation}
To conclude the proof, two other observations are necessary:
\begin{enumerate}
    \item Instantiate $X$ above as the antisymmetric subspace projector $\Pi_d^-$. Then, the expression $\tr[(R_\pi \otimes R_\sigma)(\Pi_d^-)^{\otimes t} ]$ yields a univariate polynomial in $d$.
    \item Extending the resulting polynomial to \emph{negative} formal dimension $d < 0$ yields the same polynomials but evaluated on \emph{symmetric} subspace projectors.
\end{enumerate}
These two steps are what allow us to apply standard polynomial lower bounds to obtain the $\QMA$
lower bound. As mentioned in the introduction, moving from Step 1 to Step 2 above also requires the
somewhat subtle final choice to set the $X$ as \emph{two} copies of the antisymmetric subspace
projectors (in the proof, we use the more restrictive choice $\Lambda^4 \CC^d \subseteq \Lambda^2
\CC^d \otimes \Lambda^2 \CC^d$ to obtain a smaller product value). Nonetheless, it seems the primary missing
ingredient towards a full oracle separation was instantiating the construction of
\cite{SY22} with the symmetric and antisymmetric subspaces. Thus, we believe that subsequent to the
publication of \cite{SY22} most of the technical ingredients to construct the oracle separation
were present. However, putting together the various ingredients required familiarity with~\cite{SY22}, one's membership in the ``Church of the Symmetric Subspace'' \cite{harrow2013}, as well as
intuition to regard the symmetric subspace as a ``negative'' dimension antisymmetric subspace
\cite{penrose1971applications}.

Although these ingredients (particularly the first two) are not unknown in quantum information,
connecting these pieces to formulate a complete proof is precisely a strong point of
generative AI models. We speculate that this contributed to the relative ease with which ChatGPT 5.6
Sol was able to produce the correct scaffolding of a proof with limited interaction.

\paragraph{What did we learn. } It is not always straightforward to interpret the meaning of an oracle separation, especially a quantum oracle separation~\cite{agarwal2026nonstandard}.  We can split our discussion into three parts: lessons that would be true of any quantum oracle separation between $\QMA(2)$ and $\QMA$, those that we learn from having arrived at our separation using the polynomial method, and those that are specific to our construction.

Rather than being definitive evidence that $\QMA(2)$ is not $\QMA$, quantum oracle separations should be interpreted more as ``barriers'' to certain reductions from $\QMA(2)$ to $\QMA$.  One way to interpret the canonical $\QMA(2)$ complete problem (sparse separable Hamiltonian~\cite{chailloux2012complexity}) is as follows: Given an order $4$ tensor $T$ with entries $T_{i, j, k, \ell}$, determine its injective $2$-norm, that is $\max_{\norm{\psi}_{2} = 1} \langle T , \psi^{\otimes 4}\rangle$~\cite{BFHKSZ2012}. We can also talk about the flattening of the tensor $T$ to a matrix, $T_{\mathrm{flat}} = \sum_{i, j, k, \ell} T_{i, j, k, \ell} \ket{i, j}\!\!\bra{j, \ell}$.  
In some sense, separating $\QMA(2)$ from $\QMA$ relative to a quantum oracle tells us that finding the injective $2$-norm of $T$ is extremely hard, \emph{if} you only get access to $T$ through its matrix flattening, $T_{\mathrm{flat}}$.  

To better understand what this means, it is useful to think about what our proof says about the differences between matrices and tensors.  
At a high level, our proof exploits the fact that matrices can be diagonalized into eigenvectors, which gives a low-degree approximation to the spectral norm (via the trace power method, or in-place amplification) that does not exist for the injective $2$-norm.  
This means that algorithms that get oracle access to the matrix flattening are (in some sense) not sensitive enough to changes in the entries of the matrix, whereas algorithms solving the sparse separable Hamiltonian problem must be extremely sensitive to certain changes in the structure of the underlying tensor.   
In some sense, this is not a new phenomenon.  The idea that the injective $2$-norm does not have a low-degree polynomial approximation appears in many works on problems like tensor PCA and sum-of-squares algorithms~\cite{hopkins2015tensor}. 
We also point out that this oracle separation very much mirrors the seminal work of Aaronson and Kuperberg, who provided the first quantum oracle separation between $\QMA$ and $\mathsf{QCMA}$ \cite{AK07}.  
In their work, the authors argued (via a geometric argument) that polynomial-sized classical strings are not sensitive enough to changes in the entries of a quantum state to describe a random state. 
While their work was extremely influential and impressive, we believe that both their work and ours serve to formalize and recast folklore knowledge in the language of quantum complexity theory, rather than teach us something fundamentally new about the nature of quantum proofs.  

One additional lesson from the construction of this paper is that it pinpoints a very specific way in which tensor and matrix structure can differ.
In particular, this paper shows that switching between symmetric and antisymmetric subspace projectors is a smooth transition in terms of the matrix (in that there is a single low-degree polynomial that interpolates between them), but yields a very sharp transition when viewed as a tensor.
This surprising fact has been noted in other works on mathematical physics (see, for instance \cite{king1971dimensions, penrose1971applications,dunne1989negative, keppler2023duality,elvang2005diagrammatic}, among others), but to the best of our knowledge has not been noticed by the computer science community. As a result, this observation gives a potentially exciting way to solve problems related to higher-order tensors. 

\paragraph{Future directions. } We believe this work is only the beginning of a long investigation into the nature of un-entanglement and unitary property testing.  We hope that researchers continue to be interested in understanding the complexity class $\QMA(2)$, and want to propose a few directions that we believe will lead to interesting insights about it.
\begin{enumerate}
    \item \textbf{A classical oracle separation?} Perhaps the most natural follow-up to this line of work is to search for a classical oracle separation.  This work heavily exploits the sensitivity of quantum algorithms to the entries of a matrix, when viewed as complex numbers.  However, similar to work on $\QMA$ versus $\mathsf{QCMA}$, classical oracle access to a tensor is fundamentally different than quantum oracle access to it.  In particular, while changing a single entry of the matrix would lead to a small change in the underlying matrix, it leads to a huge change in the classical description of the matrix, meaning that a classical oracle separation cannot exploit the same kinds of arguments used in this work.
    
    We note that the final resolution of the $\QMA$ versus $\QCMA$  problem ultimately relied on a completely different idea (based on the sampling reduction of \cite{BHNZ26}) than the quantum oracle separation (conditioning on a ``popular" witness) of \cite{AK07}. Indeed, many follow-up works focused on dequantizing \cite{AK07} directly and subsequently ran into technical difficulties. The same may hold true for the $\QMA$ versus $\QMA(2)$ question: it is clear that the classical oracle separation cannot use a lot of the powerful invariant theory that was leveraged here, and our oracle construction does not have an obvious classical analogue. We thus believe that the problem of finding a classical oracle separation between $\QMA(2)$ and $\QMA$ is a fundamentally challenging and interesting problem, and likely to require completely different ideas than those presented in this work. 

    Given the history of the $\QMA$ versus $\QCMA$ problem~\cite{BHNZ26,BHV26}, we speculate that examining the differing operational characteristics of $\QMA$ and $\QMA(2)$ may prove useful for separating the two classes relative to a classical oracle (or even in the unrelativized setting). In addition to the fact that $\QMA$ enjoys a strong form of in-place soundness amplification that seems to be missing from $\QMA(2)$, we note that prior literature \cite{ABD+08,BT12} has shown that $\QMA(2)$ allows for the use of surprisingly succinct proofs (in a way that we suspect should not hold for $\QMA$). Exploiting these distinguishing behaviors seems crucial both to dequantizing our separation and to understanding the broader relationship of $\QMA(2)$ with other complexity classes.
    \item \textbf{More unitary property testing? } Another interesting direction is to return to the original inspiration behind \cite{SY22}, and study other problems related to unitary property testing.  The construction from this paper provides a very interesting way to smoothly transform (locally-invariant) matrices across a sharp boundary, and may be useful for understanding other problems in unitary property testing.  Concretely, \cite{SY22} leaves open the problem of finding tight bounds for the recurrence time and entropy estimation problems, as well as finding interesting problems related to unitaries that are invariant under conjugation by various groups (discussed in \cite{Pro76}).  
    \item \textbf{State property testing? } One can formulate an analogue of $\QMA$ and $\QMA(2)$ for state property testing; these are \textsf{propQMA} and $\mathsf{propQMA}(2)$ \cite{jeronimo2024coherence}. In addition to copies of the target state, you are also provided with unentangled proof states. A single proof is not helpful, as the de-Merlinization of \cite{aaronson2005qma} can be used to remove a single proof state, but this argument does not work for $k > 1$. Could the techniques in this paper be used to exhibit a property which has exponential copy-complexity, yet is efficiently testable given $k \geq 2$ copies of a proof? As a first attempt, one could consider the mixed states proportional to the symmetric and antisymmetric subspace projectors. Interestingly, the situation is a reversal of our oracle separation, where the $\QMA(2)$ algorithm is trivial but the $\QMA$ lower bound requires the most work. For property testing, the normalization makes the \textsf{propQMA} lower bound straightforward, but there is no longer a clear $\textsf{propQMA}(k)$ algorithm for this task.
    \item \textbf{UniqueQMA and approximate counting?} Using the same proof techniques as in this paper, one can show a \emph{quantum} oracle separation between $\QMA$ and $\mathsf{uniqueQMA}$, with the northstar problem in this setting being a classical oracle separation between $\QMA$ and $\BQP^{\mathsf{uniqueQMA}}$. This is not a quantitative improvement over \cite{AHHN24} which exhibits a classical oracle separation between $\QMA$ and $\mathsf{uniqueQMA}$, but the proof is simpler. 
    We consider the ``no'' case to be the $0$-dimensional symmetric subspace, and the ``yes'' instances to be symmetric subspaces of dimension at least $2$ (this is a subfamily of the \textsf{EmptyNonEmpty} oracle problem from Anshu et al.). There is one subtlety: the polynomial constructed for the OR bound corresponds to the \emph{trace} of the acceptance operator. This means in the ``yes'' case, the value could be exponentially large. The proof of \Cref{thm:qma-lower-bound} resolves this issue by including only a single, appropriately scaled ``yes'' case, but this does not work for $\mathsf{uniqueQMA}$. Instead, the unique promise is crucially used to bound the trace by a constant; this is the differentiating feature between $\QMA$ and $\mathsf{uniqueQMA}$, which yields the separation. The work of \cite{AHHN24} also constructs oracle separations between approximate counting and $\QMA$, as well as $\QMA^\QMA$. Interestingly, approximate counting is related to a \emph{classical} oracle separation between $\QMA$ and $\QMA(2)$; more precisely, if one could show that approximate counting was in $\QMA(2)$, this, combined with the existing $\QMA$ lower bounds \cite{AKKT19} would yield a classical oracle separation between $\QMA$ and $\QMA(2)$! Thus, understanding the precise relationship between approximate counting and $\QMA(2)$ is an intriguing line of research.
\end{enumerate}

\printbibliography

\appendix

\section{Diagonalization argument}

We use the query lower bound to construct the oracle, allowing us to prove \cref{thm:oracle-separation}.
The idea is standard: enumerate all polynomial-time $\QMA$ oracle verifiers and defeat them one at a time. 
At stage $i$, we choose a fresh input length $n_i$ and use \cref{thm:qma-lower-bound} to choose the $n_i$-th oracle component on which the $i$-th verifier fails.

\oraclesep*

\begin{proof}
We will construct an input-length-indexed unitary oracle $\calO=(\calO^{(n)})_{n\geq1}$.
The $n$-th component acts on one control qubit and two $n$-qubit registers.
For a projector $P_n$ on $\C^{2^n}\otimes\C^{2^n}$,
we set
\[
\calO^{(n)} = \calO_{P_n} = \proj{0}\otimes I + \proj{1}\otimes(I-2P_n)\,.
\]
We now construct the projectors $P_n$ together with a unary language
$L^\calO$.

Enumerate all polynomial-time $\QMA$ oracle verifiers $V_1, V_2, \dots$ together with explicit polynomial bounds on their running time. 
For each $V_i$, let $m_i(n)$ and $q_i(n)$ be polynomial bounds on its witness length and number of oracle queries on input $1^n$, respectively.
We also let $a_i(n)$ be a polynomial bound on the largest oracle component that $V_i$ can query on input $1^n$. Such a bound must exist because the $r$-th oracle component itself acts on $2r+1$ qubits.

We construct the oracle in stages. 
After stage $i$, we maintain a number $B_i$ such that every oracle component of index at most $B_i$ has been permanently fixed. 
Initially, set  $B_0=0$.
Suppose stages $1,\ldots,i-1$ have been completed. We first choose a fresh
length $n_i>B_{i-1}$.
By \cref{thm:qma-lower-bound}, there is a constant $C>0$ such that any
$\QMA$ verifier for $\AES_{2^n}$ with completeness $2/3$ and soundness
$1/3$, using an $m$-qubit witness and $q$ oracle queries, satisfies
\[
q^2(m+1)\geq C \sqrt{2^n}
\]
for all sufficiently large $n$.
We choose $n_i$ sufficiently large that
\[
q_i(n_i)^2\bigl(m_i(n_i)+1\bigr) < C N(\sqrt{2^{n_i}})\,.
\]
where $N(\sqrt{2^n})$ is the function defined in \Cref{eq:n_choice}. Such a choice is always possible because the left-hand side grows polynomially
in $n_i$, whereas the right-hand side grows exponentially in $n_i$. 

We now fix all other oracle components that $V_i$ could access on this input. Set
\[
B_i \coloneqq \max\Braces{ B_{i-1}, a_i(n_i), n_i }\,.
\]
For every still-undefined component of index at most $B_i$, except for
$n_i$ itself, set $P_n=0$, so that $\calO^{(n)}=I$.
At this point, every oracle component that $V_i$ can query on input
$1^{n_i}$ has been fixed except $\calO^{(n_i)}$.

Thus, with all other queried oracle components fixed, $V_i$ on input
$1^{n_i}$ becomes a verifier for the promise problem
$\AES_{2^{n_i}}$. It remains only to choose the projector
$P_{n_i}$.
By our choice of $n_i$ and \cref{thm:qma-lower-bound}, $V_i$ cannot satisfy completeness $2/3$ and soundness $1/3$ for every
$ P\in \calY_{2^{n_i}} \cup \calN_{2^{n_i}}$.
Hence there exists some such projector $P_{n_i}$ on which $V_i$ fails.
Fix $P_{n_i}$ to be such a projector, and define
\[
1^{n_i}\in L^\calO \qquad\Longleftrightarrow\qquad P_{n_i}\in\calY_{2^{n_i}}\,.
\]
This completes the construction at the $i$-th stage.
Note that all oracle components of index at most $B_i$ are now permanently fixed.
Since $n_i>B_{i-1}$ and $B_i\geq n_i$, we have $B_i\to\infty$.
Thus every oracle component is eventually fixed.
Since $V_i$ cannot query any component of index larger than
$a_i(n_i)\leq B_i$ on input $1^{n_i}$, no later stage can change its
computation on this input. Thus $V_i$ remains incorrect on $1^{n_i}$ in
the final oracle.
We reiterate that the promise problem $\AES_D$ is only used at a sequence of selected input lengths. At all other lengths we set $P_n=0$ and declare $1^n$ to be a NO instance of the resulting unary language.

To complete the oracle, we repeat this construction for every $i$. 
By construction, for every $i$, the verifier $V_i$ fails to decide
$L^\calO$ on input $1^{n_i}$. 
Since every polynomial-time $\QMA$ oracle verifier appears in the enumeration, no such verifier decides $L^\calO$. Therefore,
\[
L^\calO\notin\QMA^\calO.
\]

All that remains is to show that $ L^\calO\in\QMA(2)^\calO$.
Following \cref{prop:qma2-upper-bound}, there is a single uniform verifier. 
On input $1^n$, the two provers each send an $n$-qubit state, and the verifier uses one query to $\calO^{(n)}$ to perform the measurement $\{P_n,I-P_n\}$, accepting on outcome $P_n$.

If $1^n\in L^\calO$, then $P_n$ is a rank-one product projector, so the provers can make the verifier accept with probability $1$. 
If $1^n \notin L^{\calO}$, then either   
\[
P_n\in\calN_{2^{n}} \qquad \text{or} \qquad P_n = 0\,.
\]
In the former case, \cref{prop:qma2-upper-bound} implies that every product witness is
accepted with probability at most $1/3$. 
In the latter case, the verifier rejects with probability $1$.
Thus, we have established 
\[
L^\calO\in\QMA(2)^\calO\,.
\]
Combining with $L^\calO \not\in \QMA^{\calO}$, we conclude that
\[
\QMA^\calO\neq\QMA(2)^\calO\,.\qedhere
\]
\end{proof}

\section{Improved disentangler bounds}
\label{sec:improved_disentangler}

We apply parallel amplification to the $\QMA(2)$ verifier to obtain a (admittedly loose) lower bound for arbitrary $\epsilon$, $\delta$.
\begin{lemma}
\label{lem:qma2_amplification}
    For every $s \leq 1$ and $k = \lceil\log_{4/3}(1/s)\rceil$, there is a $\mathsf{QMA}(2)$ verifier that makes $k$-queries to an oracle, takes as input a pair of un-entangled witnesses each of dimension $D^{k}$, and has completeness $1$ and soundness $s$. 
\end{lemma}
\begin{proof}
    The amplified $\mathsf{QMA}(2)$ verifier takes a witness $\ket{\psi}_{\reg A_1 \dots \reg A_k} \otimes \ket{\phi}_{\reg B_1 \dots \reg B_k}$, where each register $\reg A_i, \reg B_i$ has dimension $D$. Then, it uses its oracle to measure a projector $P$ $k$-many times on each pair $\reg A_i \reg B_i$.  This verifier makes $k$ oracle queries and takes as input a pair of registers of dimension $D^{k}$, as in the lemma statement.

    In the YES case, the separable value of a tensor product of many un-entangled subspaces is still $1$, so there is a proof that causes the verifier to accept with probability $1$.  In the NO case, we note that similar to the proof of \Cref{lem:product-value}, the projector $\Pi_{d}^{-, \otimes k}$ can be written as $\Pi_{d}^{-, \otimes k} (\Pi_{\mathrm{anti}}^{d, 2})^{\otimes 2k}$.  Lemma $12$ from \cite{christandl2012entanglement} shows that the separable value of $n$ copies of the antisymmetric subspace has a $h_{\mathrm{Sep}} \leq (3/4)^{n/2}$.  Setting $n = 2k$, and setting $k = \lceil\log_{4/3}(1/s)\rceil$ gives us the desired soundness.
\end{proof}

\disentangleramp*
\begin{proof}
    We first apply \Cref{lem:qma2_amplification} to the $\mathsf{QMA}(2)$ verifier from \Cref{prop:qma2-upper-bound} to get a $\mathsf{QMA}(2)$ verifier for the $\mathsf{AES}_{D}$ problem that has soundness $\eta / 2$ and completeness $1$.  Since our original verifier took a witness with dimension $D$, this new witness will be dimension $D^{\lceil \log_{4/3}(2 / \eta) \rceil}$, and the algorithm will make $\lceil \log_{4/3}(2 / \eta) \rceil$ queries.  Now we set $D' = D^{1/\lceil \log_{4/3}(2 / \eta) \rceil}$, so that the solving the $\mathsf{AES}_{D'}$ problem requires a witness of size $D$.  

    From \Cref{lem:qma_disentangler_completeness_soundness}, we now have that applying the disentangler to the $m$-qubit witness for the $\mathsf{QMA}$ verifier yields a $\mathsf{QMA}$ algorithm for the $\mathsf{AES}_{D'}$ problem with completeness-soundness gap $\eta / 2$ (which is a constant), making $O(\log(1/\eta))$ queries to the oracle.  Using \Cref{thm:qma-lower-bound}, we have the bound of
    \begin{equation*}
        \log(1/\eta)^2 (m+1) = \Omega(\sqrt{D'})\,.
    \end{equation*}
    Substituting $D' = D^{1/\lceil \log_{4/3}(2 / \eta) \rceil}$, we have that 
    \begin{equation*}
        m = \Omega\left(\frac{1}{\log(1/\eta)^2} D^{1/\lceil \log_{4/3}(2 / \eta) \rceil}\right)\,,
    \end{equation*}
    as desired.
\end{proof}

\end{document}